\documentclass[prb,preprint]{revtex4-1}

\usepackage{amsmath}
\usepackage{amsthm}
\usepackage{amsfonts}
\usepackage{graphicx}
\usepackage{url}
\usepackage[capitalise]{cleveref}
\usepackage{bm}
\usepackage{bbm}
\usepackage{subcaption}
\usepackage{ragged2e}
\newcommand{\rep}{\sigma}
\newtheorem{prop}{Proposition}
\crefname{prop}{Proposition}{Propositions}
\newtheorem{lemma}{Lemma}
\newcommand{\R}{\mathbb{R}}
\newcommand{\E}{\bm{E}}
\newcommand{\x}{\bm{x}}
\newcommand{\bb}{\bm{b}}
\newcommand{\B}{\bm{B}}
\newcommand{\intd}{\,\mathrm{d}}

\begin{document}

\title{Symmetry Promotion in Electromagnetism}

\author{Hunter Swan}
\email{orswan@stanford.edu}
\author{Jason M. Hogan}
\email{hogan@stanford.edu}
\affiliation{Department of Physics, Stanford University, Stanford, California 94305, USA}

\date{\today}

\begin{abstract}
A symmetric distribution of charges or currents typically produces an electromagnetic field having the same symmetry group.  However, under special circumstances we show that a discrete symmetry of the source distribution leads to a larger continuous symmetry for the resulting field in restricted regions.  The conditions for this promotion from discrete to continuous symmetry are that the fields be affine linear and the symmetry of the source distribution act irreducibly on the ambient space.  We illustrate this effect with many examples.  This phenomenon is pedagogically interesting both as a source of inspiration in electro- and magnetostatics and as an accessible example of group representation theory.
\end{abstract}

\maketitle 

\section{Introduction}
Symmetry is ubiquitous in modern physics, both as a foundational concept and as a practical tool for problem solving.  Historically this was not always so, and one of the first to appreciate the power of symmetry in physics was Pierre Curie, who in 1894 formulated a heuristic which has come to be known as Curie's Principle \cite{Curie}.  This observation, which originated from Curie's studies in crystallography, states that ``\textit{When certain causes produce certain effects, the symmetry elements of the causes must be found in their effects.}''\cite{castellani2016curie} 

An interesting feature of Curie's Principle (noted by Curie in the same work) is that the converse is generally not true.  That is, it may happen that an effect possesses greater symmetry than its cause.  Curie himself cited examples of this, such as the Kerr effect, wherein planar rotation symmetry of a uniform electric field may induce a change in refractive index of a medium which is symmetric both under planar rotations \textit{and} reflections through the same plane.  Such examples are nevertheless generally uncommon---in most situations causes and effects possess the same symmetry.  In this work, however, we show that there is a large class of interesting examples in electro- and magnetostatics where the symmetries of cause and effect differ. 

The example which inspired the present work is shown in \cref{fig_cylinders}, and originated as a problem in an undergraduate electromagnetism class.  It consists of three superposed charge cylinders of equal radius and uniform charge density.  The axes of the cylinders coincide with the coordinate axes in $\R^3$, so that the charge distribution has the rotational symmetry group of the cube.  (This group is known as the octahedral group, $O_h$, though details of such groups will not play an important role in our analysis.)  The electric field of a single cylinder of radius $R$, uniform charge density $\rho$, and with its axis along the $z$ coordinate axis, is given by 
\begin{equation*}
\bm{E}(x,y,z) = 
\begin{cases}
\frac{\rho \left(x \hat{x} + y\hat{y}\right)}{2\epsilon_0} & x^2+y^2 \leq R^2 \\
\frac{\rho R^2 \left(x \hat{x} + y\hat{y}\right)}{2\epsilon_0\left(x^2 + y^2\right)} & x^2+y^2 \geq R^2
\end{cases}
\end{equation*}
where $\epsilon_0$ is the permittivity of free space and we use SI units. The field of the three superposed cylinders \textit{in the region where they all overlap} is thus given by
\begin{equation*}
\bm{E}(x,y,z) = \frac{\rho}{\epsilon_0} \left(x\hat{x} + y\hat{y} + z\hat{z}\right).
\end{equation*}
We see that this electric field is completely spherically symmetric.  The symmetry of the field has been ``promoted'' from the discrete group $O_h$ of the charge distribution to the continuous group $O(3)$ of all three dimensional rotations. 

The purpose of the present work is to investigate how far this example can be pushed.  We shall show that under broad conditions of \textit{affine linearity} of the electromagnetic field and \textit{irreducibility} of a finite symmetry group action, a charge or current distribution invariant under that symmetry group can produce a field with a larger continuous rotation symmetry group in a restricted region.  We speak of the smaller symmetry of the source being ``promoted'' to the larger symmetry of the field. 

\begin{figure}[t]
\centering
    \includegraphics[width=0.9\linewidth]{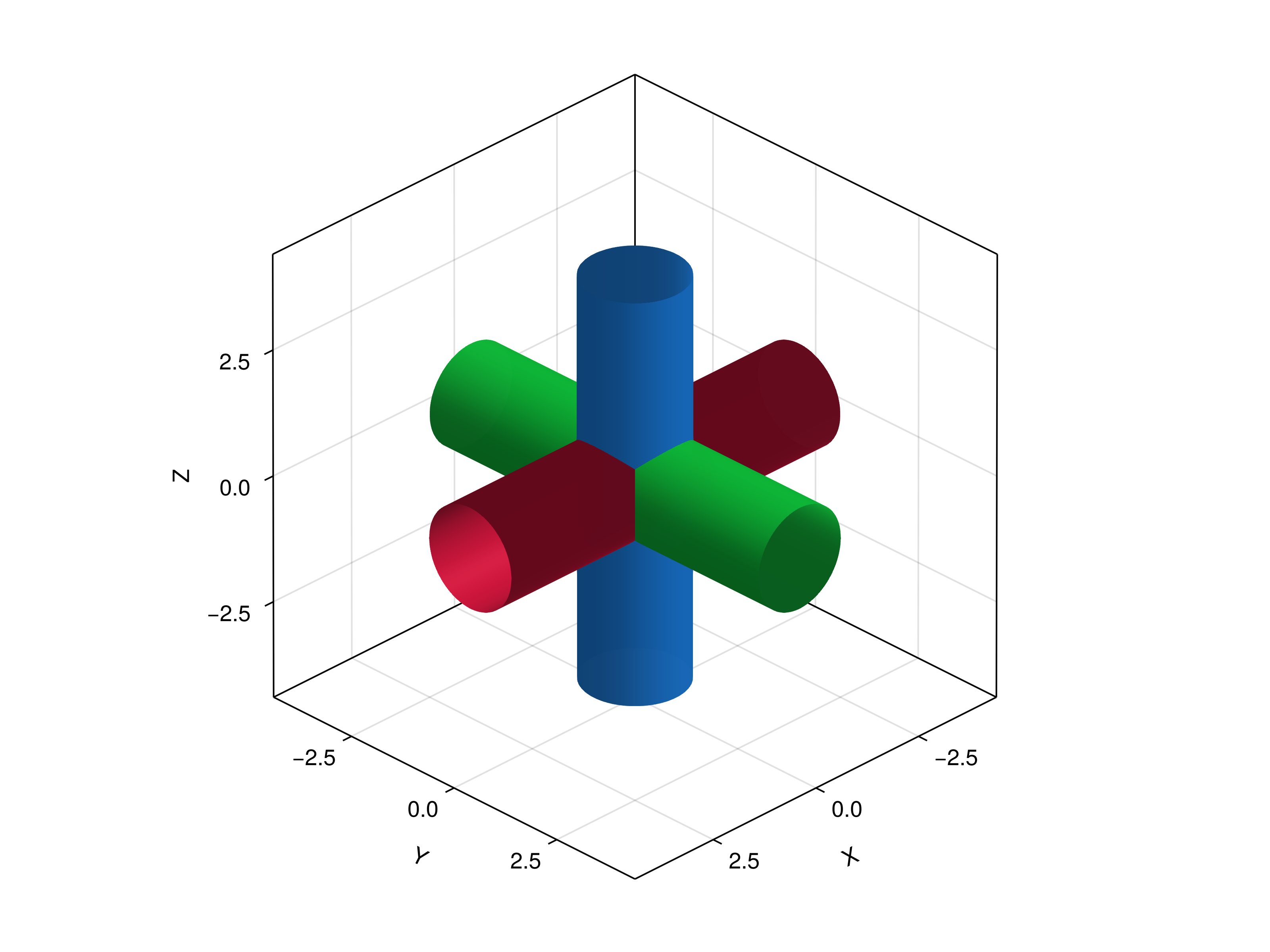}
    \caption{Three intersecting charge cylinders.  Although this charge distribution possesses only a discrete rotational symmetry group $O_h$, the electric field in the intersection of the three cylinders has a larger continuous rotational symmetry group $O(3)$.}
\label{fig_cylinders} 
\end{figure}

\begin{figure}[t] 
\centering
    \includegraphics[width=0.6\linewidth]{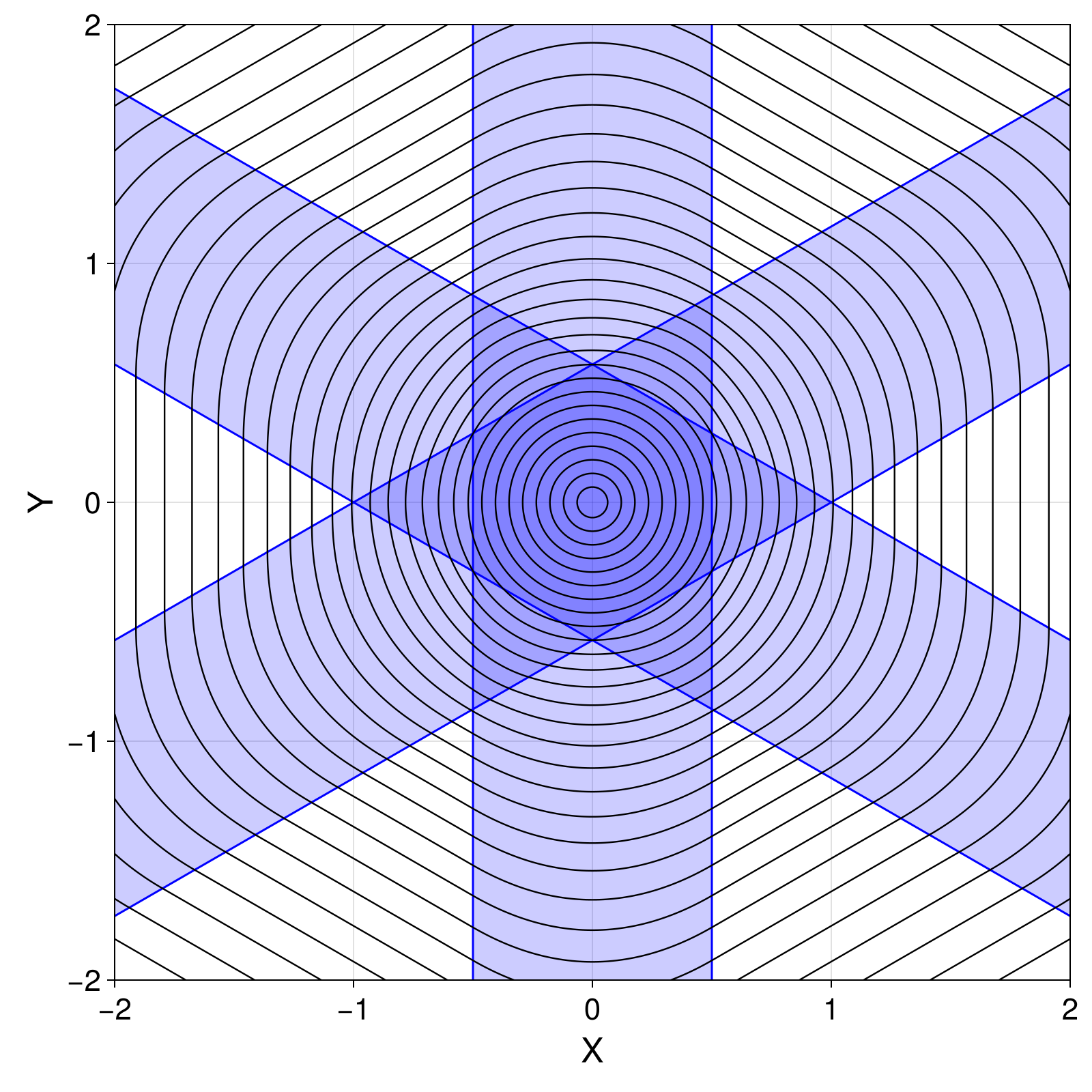}
    \caption{Superposed thick sheets of uniform current density with a sixfold rotational symmetry in the $x-y$ plane.  All current flows in the out of plane direction. Black contours are magnetic field lines.  Symmetry promotion in this case implies that the magnetic field in the region where all three sheets overlap has a complete rotational symmetry.  In particular, magnetic field lines within the central region are perfect circles, whereas outside of this region their symmetry is reduced to the symmetry of the current distribution.}
\label{fig_slabs2d} 
\end{figure}

Another example which illustrates the effect in two dimensions is provided by the thick sheets of uniform current density shown in \cref{fig_slabs2d}.  Each sheet has finite thickness $w$ in a direction $\hat{\bm n}$ in the $x-y$ plane and infinite extent in the orthogonal plane, and it carries a uniform current density $J\hat{\bm z}$.  The three sheets have normal directions $\hat{\bm n}_1 = (1,0)$, $\hat{\bm n}_2 = (-\frac{1}{2},\frac{\sqrt{3}}{2})$, and $\hat{\bm n}_3 = (-\frac{1}{2},-\frac{\sqrt{3}}{2})$ in the $x-y$ plane and are centered on the origin, so that the overall current distribution has a sixfold rotation symmetry.  In the region where the three sheets overlap, the symmetry of the resulting magnetic field is promoted from the sixfold symmetry of the system at large to the continuous rotation symmetry $O(2)$. 

An outline of this paper is as follows.  In \cref{sec_rep-theory} we introduce the necessary mathematical background.   
We have endeavored to make our exposition accessible without requiring prior familiarity with representation theory.
In \cref{sec_symmetry-promotion-3d,sec_symmetry_promotion_2d} we prove the main propositions regarding symmetry promotion in 3 dimensions and 2 dimensions, respectively.  Perhaps surprisingly, it turns out that the situation in 3D (or more generally in odd dimensions) is simpler than in 2D.  
In \cref{sec_discussion} we discuss some generalizations of the preceding results. 
We offer concluding remarks in \cref{sec_conclusion}.  

\section{Representation theory background}
\label{sec_rep-theory}
The main tool for proving statements about symmetry promotion is Schur's lemma. However, the usual statement of Schur's lemma is not quite sufficient for our needs, as it applies to complex representations, whereas we are interested in real representations. Since the precise statement of the lemma we require is not well known, and since it is not difficult to prove from first principles, we shall provide a self-contained overview. 

Given a group $G$ and a real $n$-dimensional vector space $\R^n$, a \textit{representation} of $G$ on $\R^n$ is a function $\rep:G\rightarrow \mathrm{GL}(n,\mathbb{R})$ from $G$ into the set of invertible $n\times n$ real matrices which preserves the group operation, i.e. for any two group elements $g,h\in G$ we have $\rep(gh) = \rep(g)\rep(h)$.  We speak also of the group $G$ ``acting on $\R^n$''. 
A \textit{subrepresentation} is a vector subspace $V\subseteq \R^n$ with the property that for all $g\in G$ the matrix $\rep(g)$ maps $V$ into itself. 
The representation is \textit{irreducible} if it has no subrepresentations, apart from the trivial ones consisting of the full space $\R^n$ and the zero space $\{0\}$.  We speak in this case of the group $G$ ``acting irreducibly on $\R^n$''.

The first part of Schur's lemma states that if $G$ acts irreducibly on $\R^n$ via representation $\sigma$ and if $A$ is a non-zero linear map from $\R^n$ to itself which commutes with $\rep(g)$ for all $g\in G$, then $A$ must necessarily be an isomorphism.  This follows from the fact that the kernel and image of $A$ are both subrepresentations of $\R^n$:  Given any $v\in \ker A$ and $g\in G$, we have $A \rep(g) v = \rep(g) A v = 0$, hence $\rep(g)v\in \ker A$, and thus $\ker A$ is a subrepresentation.  Similarly, given any $v\in \text{im } A$ and $g\in G$, $v = Aw$ for some $w$, hence $\rep(g)v = \rep(g) A w = A\rep(g)w$, hence $\rep(g) v \in \text{im }A$, hence $\text{im }A$ is also a subrepresentation.  Since $G$ acts irreducibly and $A$ is non-zero, $\ker A$ must be zero, hence $A$ is injective.  Then since $\text{im }A$ is non-zero and also a subrepresentation, it must be the full space $\R^n$, hence $A$ is surjective.  So $A$ is an isomorphism as claimed. 

The second part of Schur's lemma states that under certain conditions the matrix $A$ in the previous paragraph must in fact have the form $A = \lambda \mathbbm{1}$, where $\lambda\in\R$ is a scalar and $\mathbbm{1}$ is the $n\times n$ identity matrix.  For this part of the lemma to hold, it is required that the matrix $A$ have an eigenvalue.  If $A$ has an eigenvalue $\lambda$, then by definition the matrix $A-\lambda \mathbbm{1}$ annihilates some vector, i.e. there is some non-zero $v\in \ker \left(A-\lambda \mathbbm{1}\right)$.  Now if the matrix $A-\lambda\mathbbm{1}$ were not uniformly zero, then the previous paragraph would apply and we would have a contradiction.  Hence we find that $A-\lambda\mathbbm{1} = 0$, proving the second part of Schur's lemma.  

The question now becomes under what circumstances must $A$ have an eigenvalue?  The usual statement of Schur's lemma is for representations over algebraically closed fields, such as the complex numbers $\mathbb{C}$.  In this case, the fundamental theorem of algebra guarantees that the characteristic polynomial of $A$ has a root, which implies the existence of an eigenvalue.  However, on a real vector space like $\R^n$, the characteristic polynomial of $A$ also is guaranteed to have a root \textit{provided the dimension of the space is odd}.  This is because the characteristic polynomial of $A$ has degree $n$, and an odd degree polynomial always has a root.  

So we conclude that the second part of Schur's lemma applies to odd dimensional real representations, but not to even dimensional spaces.  It will turn out that this is not an academic curiosity---the failure of Schur's lemma in 2 dimensions has interesting consequences, as we shall see.  We summarize this result as 
\begin{lemma}
Let $G$ act irreducibly on $\R^n$, and let $A:\R^n\rightarrow\R^n$ be a non-zero linear map which commutes with $\rep(g)$ for all $g\in G$.  Then $A$ is an isomorphism, and if $n$ is odd, then $A = \lambda \mathbbm{1}$ for some $\lambda \in \R$. 
\end{lemma}

We conclude this section by recalling the transformation properties of fields of scalars, vectors, and pseudo-vectors under groups of rotations.
For a subgroup $G\subseteq O(n)$ of the set of rotations of $\R^n$ we have $\rep(g) = g$ for the usual representation on $\R^n$.  This naturally gives an action on scalar fields defined over $\R^n$, whereby a field $\rho(\x)$ is mapped to $\rho(g^{-1}\x)$ by an element $g\in G$.  Similarly, there is a natural action on vector fields sending $\bm{v}(\x)$ to $g\bm{v}(g^{-1}\x)$, and a natural action on pseudo-vector fields sending $\bm{v}(\x)$ to $\det(g) g\bm{v}(g^{-1}\x)$. 

In the context of electromagnetism, static distributions of charge $\rho(\x)$ or current $\bm{J}(\x)$ which are invariant under $G$ (i.e. for all $g\in G$ we have $\rho(g^{-1}\x) = \rho(\x)$ and $g\bm{J}(g^{-1}\x) = \bm{J}(\x)$) produce electric or magnetic fields which are also invariant under the same group.  This is intuitively clear, but can be derived analytically from the laws of Coulomb and Biot-Savart.  For electric fields we have 
\begin{multline}
\E(\x) = \frac{1}{4\pi\epsilon_0}\int_{\R^3} \frac{\rho(\x') \left(\x-\x'\right)}{\left|\x - \x'\right|^3} \intd^3\x' = \frac{1}{4\pi\epsilon_0}\int_{\R^3} \frac{\rho(g\x') \left(\x-g\x'\right)}{\left|\x - g\x'\right|^3} \intd^3(g\x') \\ = \frac{1}{4\pi\epsilon_0}g\int_{\R^3} \frac{\rho(\x') \left(g^{-1}\x-\x'\right)}{\left|g^{-1}\x - \x'\right|^3} \intd^3\x' = g \E(g^{-1}\x),
\end{multline}
and similarly for magnetic fields 
\begin{multline}
\B(\x) = \frac{\mu_0}{4\pi} \int_{\R^3}\frac{\bm{J}(\x') \times (\x - \x')}{\left| \x - \x' \right|^3} \intd^3\x' = \frac{\mu_0}{4\pi} \int_{\R^3}\frac{\bm J(g\x') \times (\x - g\x')}{\left| \x - g\x' \right|^3} \intd^3(g\x') \\ = \frac{\mu_0}{4\pi} g \int_{\R^3} \det(g)\frac{\bm J(\x') \times (g^{-1}\x - \x')}{\left| g^{-1}\x - \x' \right|^3} \intd^3\x' = \det(g) g\B(g^{-1}\x).
\end{multline}
We have used here the transformation properties of the cross-product $(g\bm{v})\times (g\bm{w}) = \det(g)g(\bm{v}\times \bm{w})$ and of the volume element $\intd^3(g\x)=\left|\det(g)\right| \intd^3\x = \intd^3\x$, where $\left|\det(g)\right| = 1$ for a rotation. 

\section{Symmetry Promotion in 3 Dimensions}
\label{sec_symmetry-promotion-3d}
We now establish our main result for symmetry promotion in 3D.  We will show that under appropriate conditions, affine linear vector or pseudo-vector fields which are invariant under finite groups of proper rotations are in fact of the form $\bm v(\x) = \lambda \x$ for some $\lambda$, and thus invariant under arbitrary proper rotations.  Note that it is important to allow for affine linearity, rather than just linearity, as it allows us to form composite charge distributions by arbitrary translations of elementary distributions such as the uniform sphere, cylinder, or slab.  In \cref{sec_symprom_3d_examples}, we will form superpositions of such elementary distributions which are invariant under an irreducible group action.  The following proposition will then imply that these superpositions exhibit symmetry promotion. 

\begin{prop}
Let $G$ be a finite subgroup of $SO(3)$ which acts irreducibly on $\R^3$.  Let $\bm{v}(\x)$ be a (pseudo-)vector field which is invariant under $G$ and which is affine linear on some region $U\subset \R^3$, where $U$ is also invariant under $G$ and contains at least one nonzero point $\x \in \R^3$. Then there is some scalar $\lambda \in \R$ such that $\bm{v}(\x) = \lambda \x$ for all $\x \in U$. 
\label{prop_3d}
\end{prop}
\begin{proof}
The hypothesis of affine linearity implies there exist a matrix $A$ and vector $\bb$ so that for $\x \in U$ we have $\bm{v}(\x) = A\x + \bb$.  The hypothesis of invariance under $G$ implies that for all $\x \in U$ and $g \in G$ we have
\begin{equation}
Ag\x + \bb = \bm{v}(g\x) = g\bm{v}(\x) = gA\x + g\bb.
\end{equation}
We will first show that $\bb=0$ by way of the following

\begin{lemma}
For any $\x \in \R^3$, we have 
\begin{equation}
\sum_{g\in G} g\x = \bm 0
\end{equation}
\end{lemma}
\begin{proof}
Let $\bm y:= \sum_{g\in G} g\x$.  For any $g'\in G$ we have 
\begin{equation}
g'\bm y = g' \sum_{g\in G} g\x = \sum_{g\in G} g'g\x = \sum_{g\in G} g\x = \bm y.
\end{equation}
In the penultimate equality we have used the fact that $g'g$ assumes every value in $G$ exactly once as $g$ ranges over all elements of $G$.  Hence the subspace of $\R^3$ spanned by $\bm y$ is invariant under the action of $G$.  Since $G$ acts irreducibly, $\bm y=\bm 0$. 
\end{proof}

To show that $\bb=0$, consider the quantity $\frac{1}{|G|}\sum_{g\in G} \bm v(g\x)$.  On the one hand we have
\begin{equation}
\frac{1}{|G|}\sum_{g\in G} \bm v(g\x) = A \left(\frac{1}{|G|}\sum_{g\in G} g\x\right) + \bb = \bb.
\end{equation}
On the other hand we have
\begin{equation}
\frac{1}{|G|}\sum_{g\in G} \bm v(g\x) = \frac{1}{|G|}\sum_{g\in G} g\bm v(\x) = \frac{1}{|G|}\sum_{g\in G} g A\x + \frac{1}{|G|}\sum_{g\in G} g\bb = \bm 0,
\end{equation}
and thus $\bb = \bm 0$. 

Having established that $\bm v(\x) = A\x$, we see that for all $g\in G$ and $\x \in U$ we have $Ag\x = gA\x$.  We would like to show that this holds for any $\x \in \R^3$ so that we can conclude that $Ag=gA$ and apply Schur's lemma.  

Note that for any nonzero $\x\in U$, the set of points $\{g\x \; | \; g\in G\}$ spans $\R^3$, for if not then the span of the former set would be a proper invariant subspace with respect to the action of $G$.  Hence the points of $U$ span $\R^3$, and so by linearity the relation $Ag\x = gA\x$ holds for any $\x \in \R^3$. Thus $Ag = gA$ for any $g \in G$, i.e. $A$ commutes with all elements of $G$.  Finally, appealing to Schur's lemma for odd-dimensional real vector spaces, $A=\lambda \mathbb{I}$ for some $\lambda\in\R$, and the proposition is proved. 
\end{proof}
Note that the proof works equally well replacing $SO(3)$ and $\R^3$ with any $SO(n)$ and $\R^n$ as long as $n$ is odd.  In the case of an electric field $\E(\x)$ satisfying the hypothesis of the proposition, the constant $\lambda$ may be determined from Gauss' Law, 
\begin{equation}
\lambda = \frac{1}{3}\nabla \cdot (\lambda \x) = \frac{1}{3}\nabla \cdot \E(\x) = \frac{\rho(\x)}{3\epsilon_0},
\end{equation}
which also shows that the charge density must be constant in the region of interest $U$.  The analogous result for magnetic fields shows that the only solutions satisfying the hypotheses of \cref{prop_3d} are in fact uniformly zero in the region $U$, since $\nabla \cdot \B = 0$.  We will see below that the 2D case is more interesting for magnetic fields. 

\subsection{3D Examples}
\label{sec_symprom_3d_examples}

\begin{figure}[b!]
    \centering
    
    \begin{subfigure}[b]{0.435\textwidth}
        \centering
        \includegraphics[width=\textwidth]{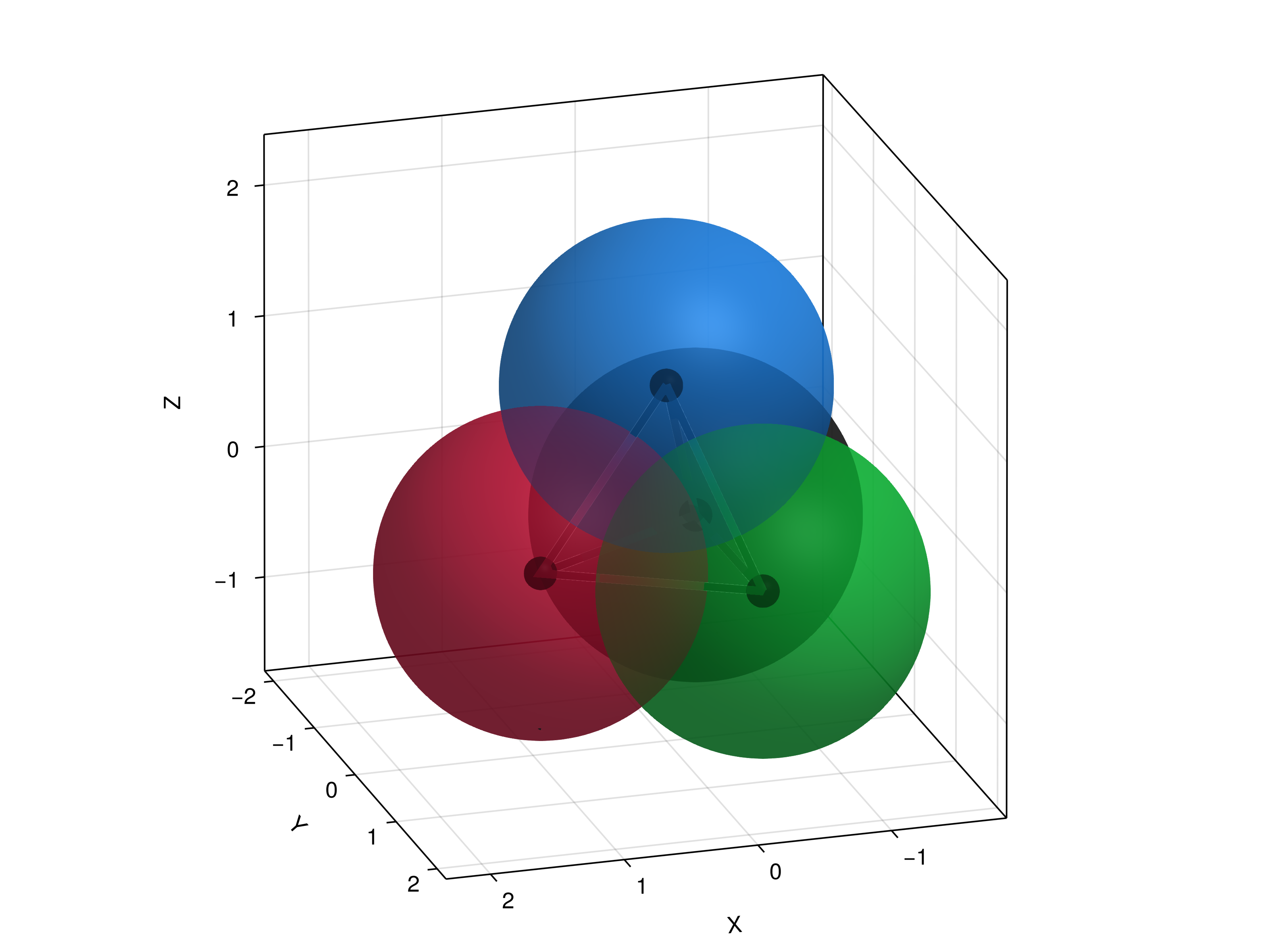}
        \caption{Spheres at the vertices of a regular tetrahedron.}
        \label{fig_3d_sub1}
    \end{subfigure}
    \hfill
    \begin{subfigure}[b]{0.435\textwidth}
        \centering
        \includegraphics[width=\textwidth]{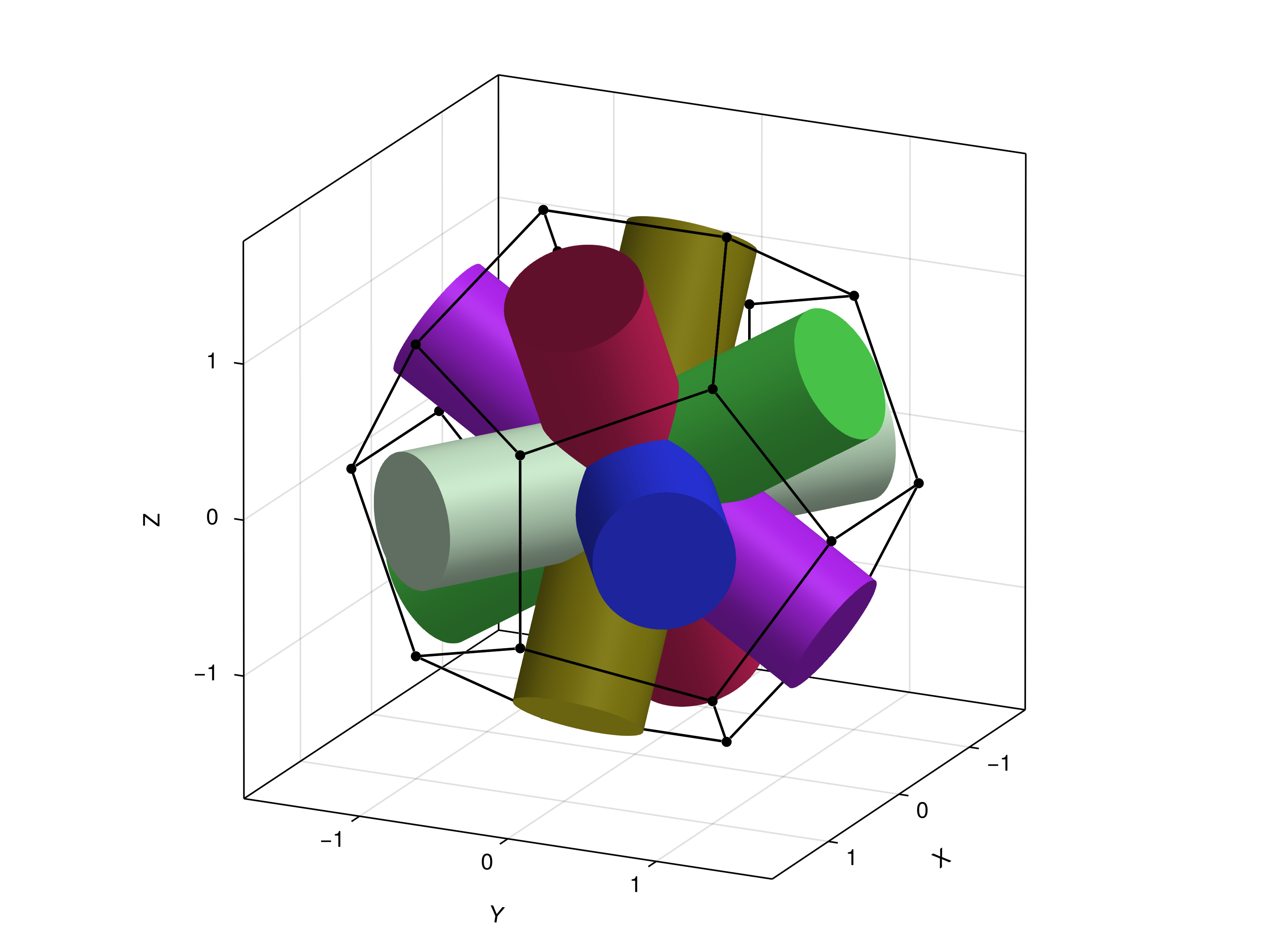}
        \caption{Infinite cylinders piercing the faces of a dodecahedron.}
        \label{fig_3d_sub2}
    \end{subfigure}
    
    \begin{subfigure}[b]{0.435\textwidth}
        \centering
        \includegraphics[width=\textwidth]{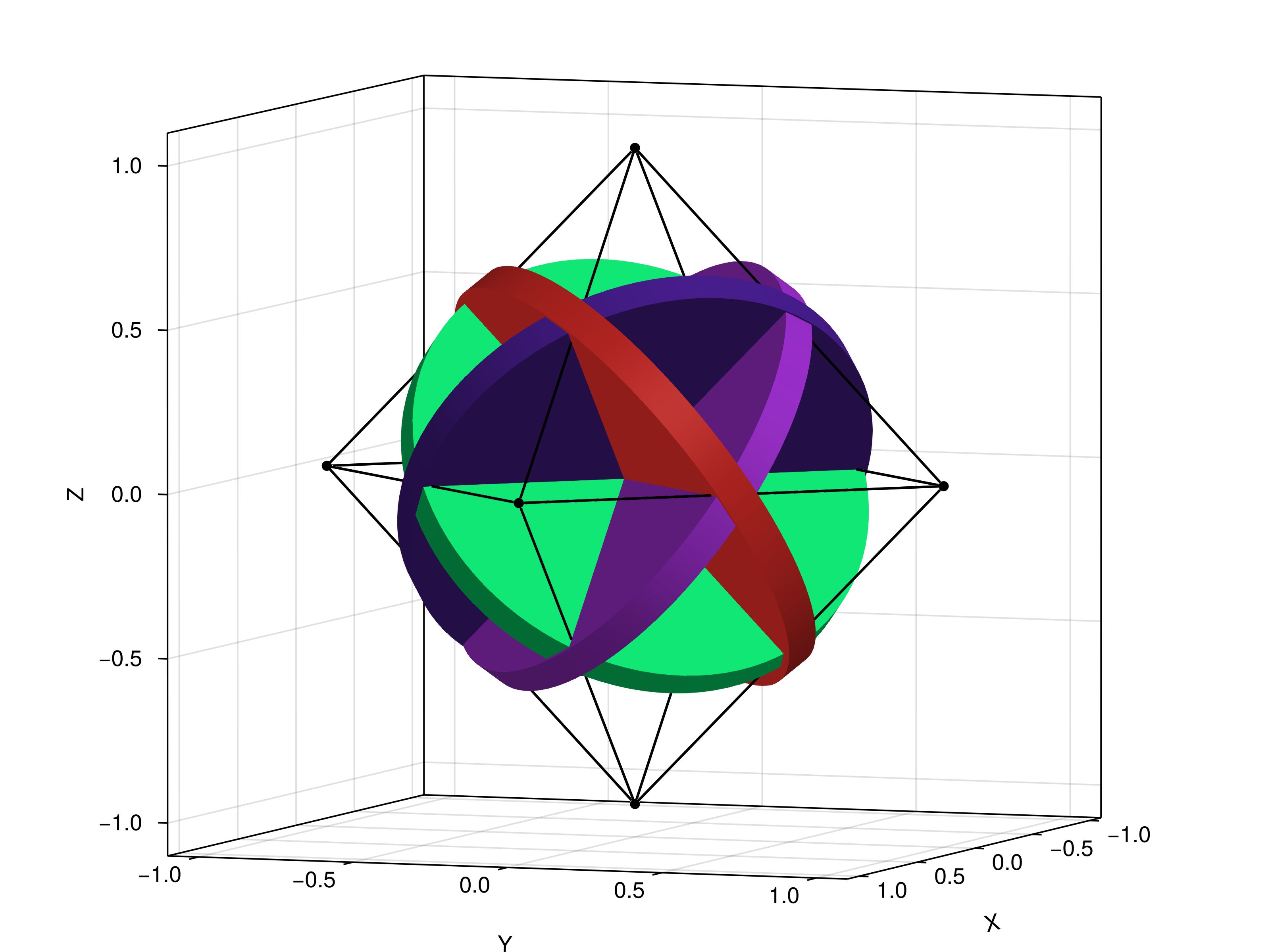}
        \caption{Infinite slabs parallel to faces of an octahedron, centered on the origin.}
        \label{fig_3d_sub3}
    \end{subfigure}
    \hfill
    \begin{subfigure}[b]{0.435\textwidth}
        \centering
        \includegraphics[width=\textwidth]{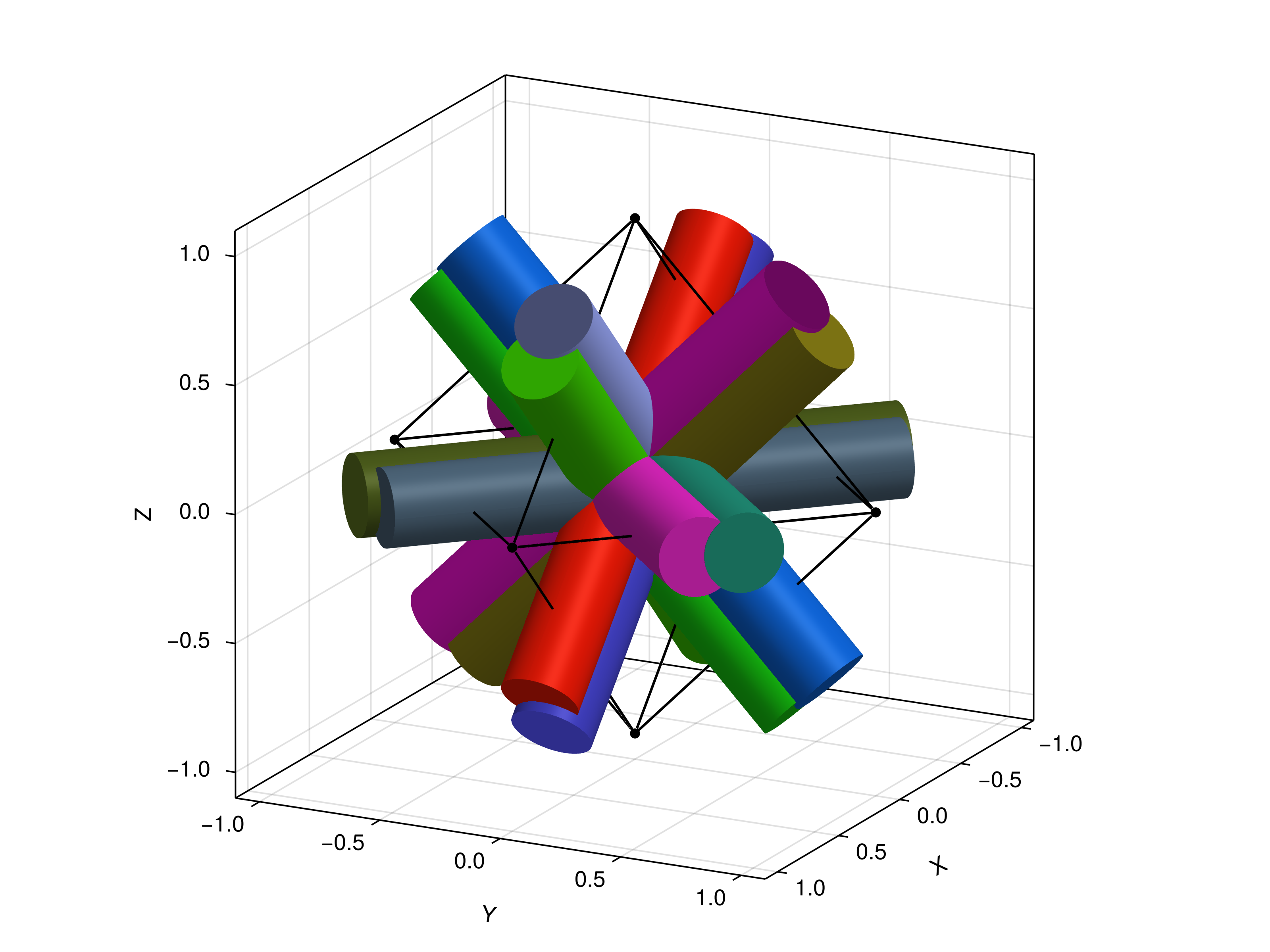}
        \caption{Infinite cylinders with axes aligned to the edges of an octahedron.}
        \label{fig_3d_sub4}
    \end{subfigure}
    
    \caption{Examples of charge distributions to which 3D symmetry promotion applies.  Each distribution is constructed by replicating an elementary distribution so as to have the symmetry of some Platonic solid.  In these examples, we have aligned the elementary distributions of the infinite slab (shown truncated as a disk), infinite cylinder (also truncated), and sphere to either the vertices, faces, or edges of a Platonic solid.  Other symmetric polyhedra (e.g. truncated Platonic solids) also suffice for this construction, provided they have a symmetry group which acts irreducibly on $\R^3$.}
    \label{fig_3d_examples}
\end{figure}

To generate examples of electric fields to which \cref{prop_3d} applies, we may begin by picking a charge distribution $\rho_0(\x)$ which produces an affine linear field near the origin and a group $G \subset SO(3)$ which acts irreducibly on $\R^3$.  We then symmetrize the charge distribution by superposing rotated distributions of the form $\rho_0(g^{-1}\x)$ for $g\in G$.  Examples of charge distributions which work for this purpose include the infinite cylinder, the infinite slab, and the ball, in each case with a uniform interior charge density.  Examples of symmetry groups which work for this purpose are those of the Platonic solids.  The example of \cref{fig_cylinders} is one example of this construction.  Note that (as the example of \cref{fig_cylinders} illustrates) to symmetrize the charge distribution, it always suffices to sum $\rho_0(g^{-1}\x)$ over all $g\in G$, but if $\rho_0$ is already symmetric under some elements $g$, then a smaller sum also suffices.  We give other examples of this form in \cref{fig_3d_examples}.

The patterns of these examples suggest two natural questions. Firstly, must the region to which symmetry promotion applies necessarily include the origin?  Secondly, are there other \textit{bounded} charge distributions besides the uniform ball which produce a linear field and can thus be employed in the above construction? 

\begin{figure}[b]
\centering
    \includegraphics[width=0.5\linewidth]{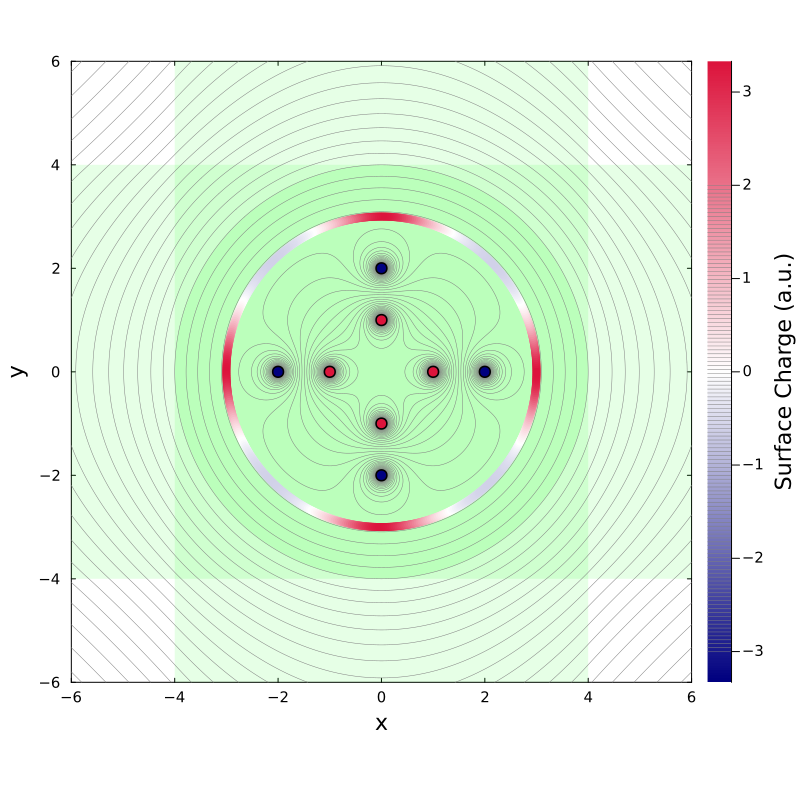}
    \caption{Equipotential curves in the $x-y$ plane of a charge distribution for which symmetry promotion occurs in a region which excludes the origin.  The coloration of the circle of radius 3 indicates the interior surface charge induced on the metal spherical shell by the point charges.  The transparent green regions are the continuous charge densities of the cylinders shown in \cref{fig_cylinders}.}
\label{fig_origin-excluded} 
\end{figure}

\begin{figure}[t]
    \centering
    
    \begin{subfigure}[b]{0.25\textwidth}
        \centering
        \includegraphics[width=\textwidth]{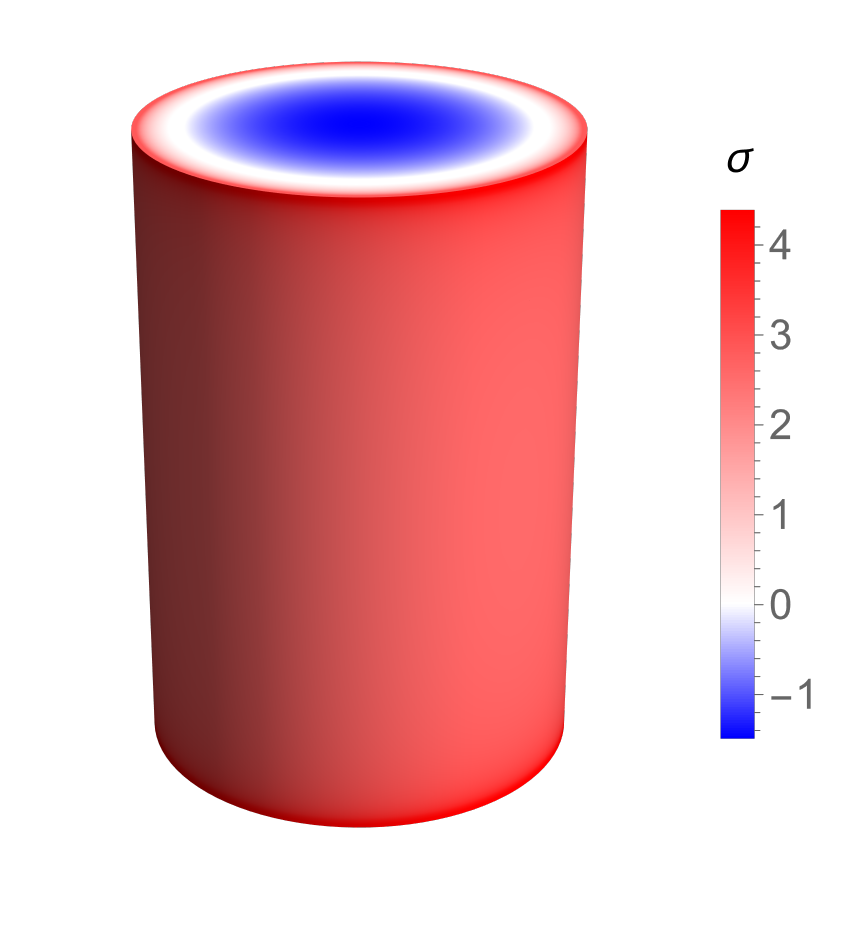}
        \caption{}
        \label{fig_charge_cyl_sub1}
    \end{subfigure}
    \begin{subfigure}[b]{0.33\textwidth}
        \centering
        \includegraphics[width=\textwidth]{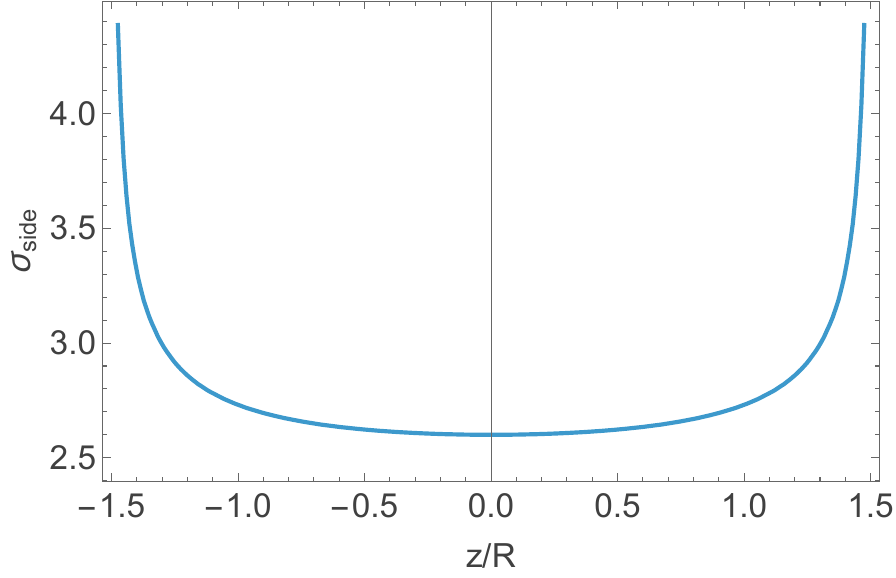}
        \caption{}
        \label{fig_charge_cyl_sub2}
    \end{subfigure}
    \begin{subfigure}[b]{0.33\textwidth}
        \centering
        \includegraphics[width=\textwidth]{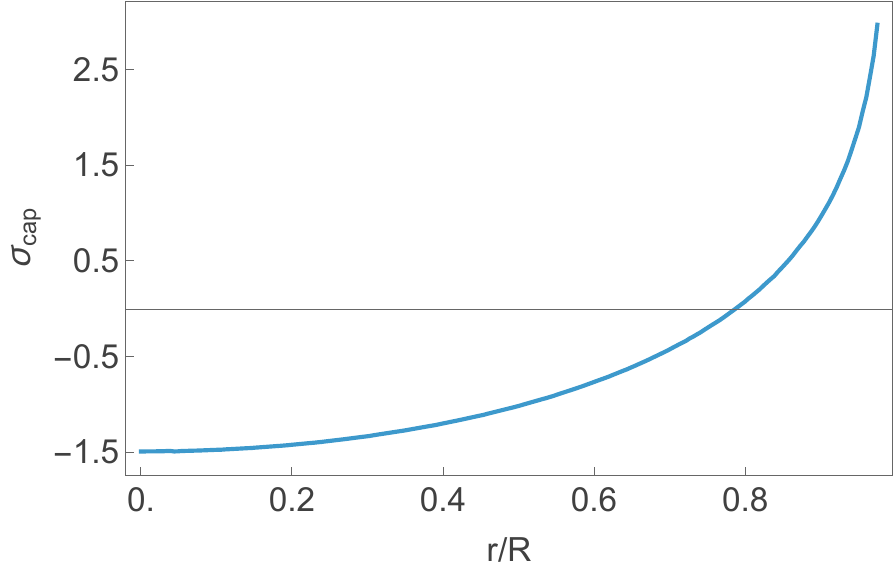}
        \caption{}
        \label{fig_charge_cyl_sub3}
    \end{subfigure}
    \caption{Surface charge distribution of a finite length cylinder with a uniform interior charge density and radial linear interior field, equivalent to that of an infinite cylinder.  All charge densities are in arbitrary units, and distances are in units of the cylinder radius.  The height of the cylinder is 3.  \textbf{(a)} shows the surface charge density via coloration on the cylinder surface.  \textbf{(b)} plots the charge density along a vertical line from the bottom to the top of the cylinder side.  \textbf{(c)} plots the surface charge density along a radial line from the center to the edge of the cylinder cap.}
    \label{fig_charge_cyl}
\end{figure}

The answer to the first question is negative.  Considering the example of \cref{fig_cylinders}, let us modify a small ball of radius $3\varepsilon$ near the origin as follows.  Place point charges of charge $+q$ at $\pm \varepsilon\hat x$, $\pm \varepsilon\hat y$, and $\pm \varepsilon\hat z$ and point charges of charge $-q$ at $\pm 2\varepsilon\hat x$, $\pm 2\varepsilon\hat y$, and $\pm 2\varepsilon\hat z$.  Then surround them with a conducting spherical shell of radius $3\varepsilon$ centered at the origin.  Inside the sphere the electric field will be highly non-linear and not spherically symmetric, but outside of the sphere the surface charge of the metal precisely cancels the field of the interior charges, as shown in \cref{fig_origin-excluded}.  Thus symmetry promotion in this case will occur in the region exterior to the sphere but interior to all of the charge cylinders. 

The answer to the second question is positive.  Examples may be generated by truncating the infinite slab or cylinder to any bounded region $B$ and applying appropriate surface charges on the boundary of this region so as to maintain the same interior electric field as for the original distribution.  To determine this surface charge distribution, we solve the Laplace equation in the region exterior to the finite distribution subject to boundary conditions of continuity across the surface $\partial B$ and the interior potential implied by the corresponding infinite charge distribution.  An example of such a distribution is shown in \cref{fig_charge_cyl}. 

\begin{figure}[t]
\centering
    \includegraphics[width=0.7\linewidth]{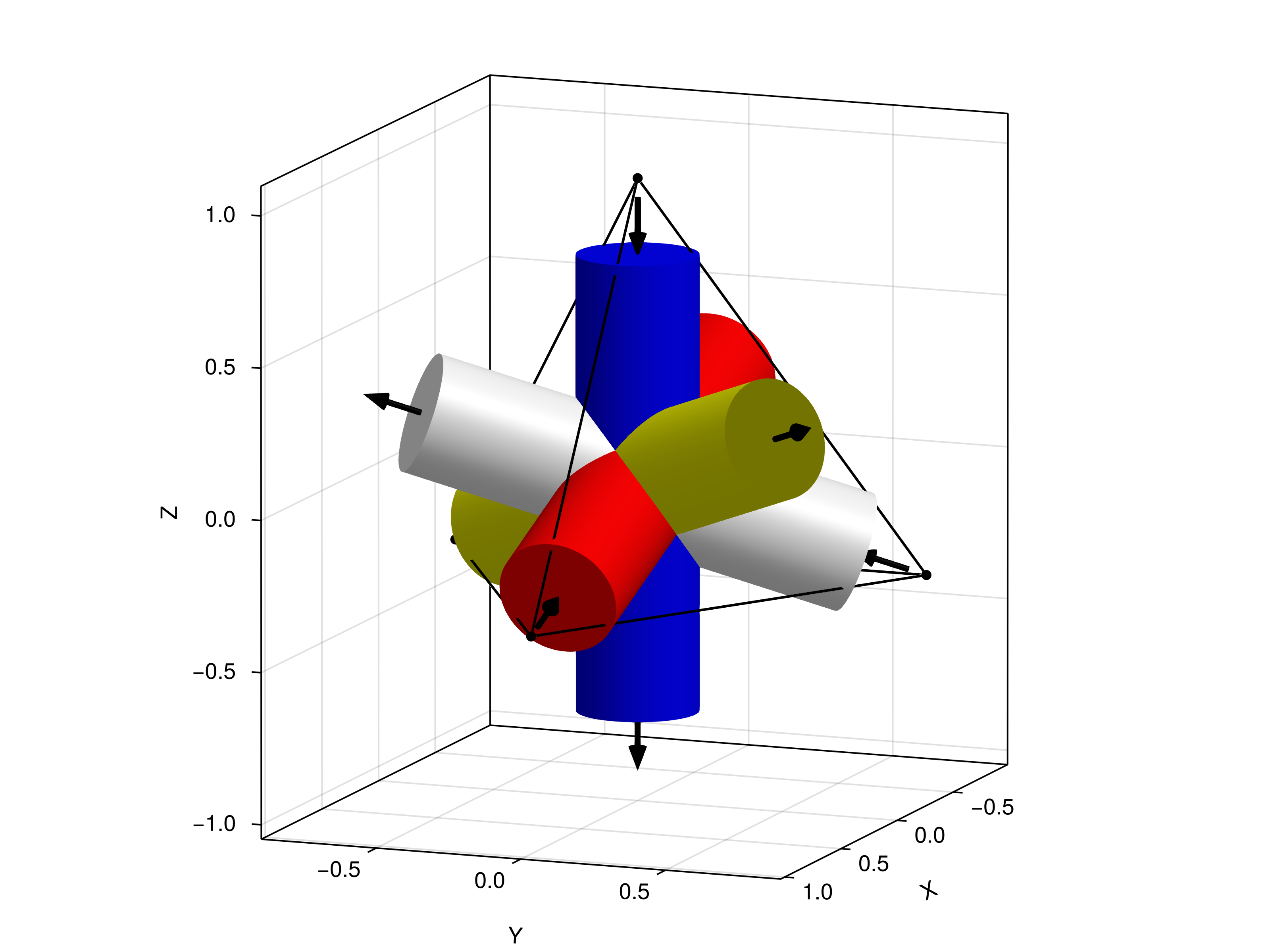}
    \caption{Current distribution of thick wires piercing the faces of a regular tetrahedron, exhibiting symmetry promotion near the origin, so that the magnetic field is zero in the region of intersection of all wires.  Arrows indicate the direction of current flow.}
\label{fig_tet_cyl}
\end{figure}

In the case of magnetic fields, a charge distribution comprised of infinite cylinders or slabs may be reinterpreted as a current distribution of thick wires or sheets, provided a direction of current flow can be assigned which preserves the symmetry group.  This requires in particular that no wire have an axis which is inverted by a symmetry operation, and no slab have an axis which is preserved by a non-trivial symmetry operation, which excludes e.g. the distributions of \cref{fig_3d_sub2,fig_3d_sub3,fig_3d_sub4}.  An example of non-trivial 3D symmetry promotion for magnetic fields is provided by piercing each face of a regular tetrahedron by a thick wire, such that positive current flows outward through the face, as shown in \cref{fig_tet_cyl}.  As noted above, for magnetic fields in 3D, the consequence of symmetry promotion is that the field is uniformly zero in the high symmetry region. 

\section{Symmetry Promotion in 2 Dimensions}
\label{sec_symmetry_promotion_2d}
In 2D, the stronger form of Schur's lemma fails in general for real representations, and with it also \cref{prop_3d}.  It turns out that symmetry promotion still holds in 2D, but in a more non-trivial way.  We formulate this as 
\begin{prop}
Let $G$ be a finite subgroup of $SO(2)$ which acts irreducibly on $\R^2$.  Let $\bm{v}(\bm{x})$ be a (pseudo-)vector field which is invariant under $G$ and which is affine linear on some region $U\subset \R^2$, where $U$ is also invariant under $G$.  Then there are some scalars $\lambda,\mu\in\R$ such that $\bm{v}(\bm{x}) = \lambda\bm{x} + \mu Y \bm{x}$ for all $\bm{x}\in U$, where $Y=\left(\begin{smallmatrix}0 & -1 \\ 1 & 0\end{smallmatrix}\right)$.
\label{prop_2d}
\end{prop}
Note that any element of $SO(2)$ can be written as $\cos(\theta)\mathbbm{1} + \sin(\theta)Y$ for some $\theta$, so \cref{prop_2d} implies that any such rotation commutes with the matrix $\lambda \mathbbm{1} + \mu Y$ characterizing $\bm{v}(\bm{x})$, and thus $\bm{v}(\bm{x})$ in the region $U$ is invariant under arbitrary rotations in $SO(2)$. 

\begin{proof}
As in the proof of \cref{prop_3d}, $\bm{v}(\bm{x})$ is in fact a linear function, $\bm{v}(\bm{x}) = A\bm{x}$ for some $A\in\R^{2\times 2}$ satisfying $Ag = gA$ for all $g\in G$. Irreducibility of $G$ also implies that there exists some $g\in G$ of the form 
\begin{equation}
g=\begin{pmatrix} \cos(\theta) & -\sin(\theta) \\ \sin(\theta) & \cos(\theta) \end{pmatrix}
\end{equation}
with $\theta$ not a multiple of $\pi$, so that the off-diagonal terms are non-zero, for otherwise the $x$-axis would be an invariant subspace of the action of $G$.  Since $A$ commutes with this $g$, it also commutes with $(g-\cos(\theta) \mathbbm{1})/\sin(\theta) = Y$.  A basis for the space of all $2\times 2$ real matrices is given by $\mathbbm{1}$ and $Y$ together with $Z = \left(\begin{smallmatrix}1 & 0 \\ 0 & -1\end{smallmatrix}\right)$ and $X = \left(\begin{smallmatrix}0 & 1 \\ 1 & 0 \end{smallmatrix}\right)$.  (As an aside, we note that $X$, $iY$, and $Z$ are the Pauli matrices.)  Writing $A = \lambda\mathbbm{1} + \mu Y + \nu X + \xi Z$ and using the easily verified identities $[Y,X] = -2Z$, $[Y,Z] = 2X$, we find 
\begin{equation}
0 = [Y,A] = 2\xi X - 2\nu Z.
\end{equation}
Since $X$ and $Z$ are linearly independent, we must have $0=\nu=\xi$, proving the proposition. 
\end{proof}

\begin{figure}[b]
\centering
    \includegraphics[width=0.6\linewidth]{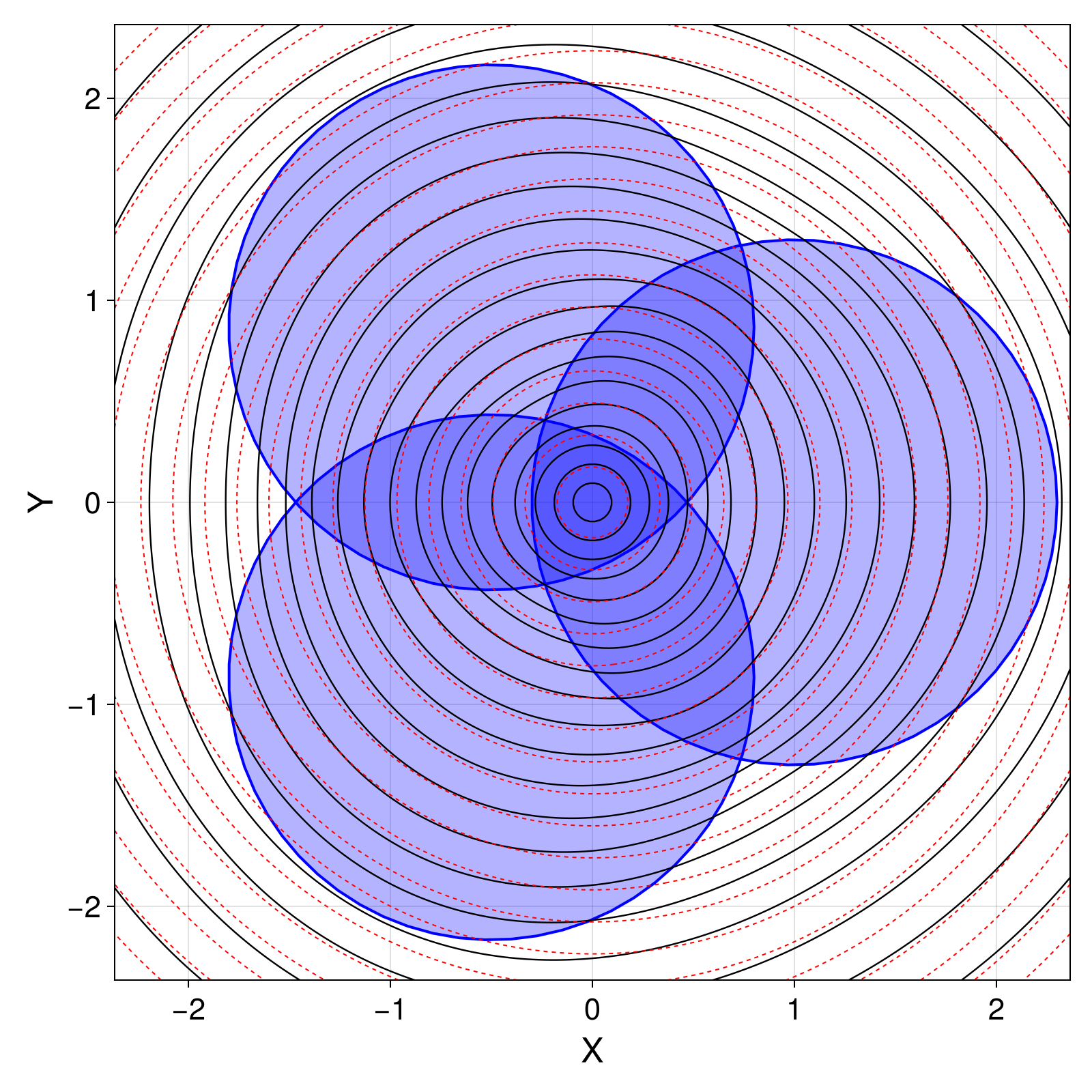}
    \caption{Current carrying wires with a threefold rotational symmetry in the $x-y$ plane.  Black contours are magnetic field lines.  Dashed red contours are circles centered on the axis of symmetry to emphasize the asymmetry of the magnetic field lines outside of the central region.  This figure may also be interpreted as a cross section of uniform charge cylinders, with the magnetic field lines being reinterpreted as equipotential curves.}
\label{fig_wires2d} 
\end{figure}

As in 3D, Maxwell's equations provide useful constraints on charge and current distributions which exhibit 2D symmetry promotion.  Note that the ($z$-component of) curl of $\bm{v}(\bm{x}) = \lambda \bm{x} + \mu Y\bm{x}$ is $2\mu$, so by Ampere's Law if a magnetic field satisfies the hypotheses of \cref{prop_2d} and there is no time varying electric field, we must have $\mu = \frac{1}{2}\mu_0 J$, where $J$ is the ($z$-component of) current density, along with $\lambda = 0$ from $\nabla \cdot \B = 0$.  Similarly, for an electric field satisfying the hypotheses of \cref{prop_2d}, provided there are no time varying magnetic fields, from Faraday's Law we must have $\mu=0$, and from Gauss' Law we have $\lambda = \frac{\rho}{2\epsilon_0 }$. 

The finite subgroups of $SO(2)$ can be identified with the $n$-th roots of unity for some positive integer $n$, i.e. the set of all rotations through angles $\frac{2k\pi}{n}$ with $k\in\{0,1,\dots,n-1\}$.  For the group to act irreducibly on $\R^2$ we must further require $n>2$.  Elementary charge distributions which generate affine linear fields in the plane include transverse cross sections of the infinite slab and cylinder.  The example of \cref{fig_slabs2d} can be reinterpreted as a charge distribution of uniform slabs, with magnetic field lines becoming equipotential curves.  Elementary current distributions which produce affine linear fields include cross sections of the thick wire and thick sheet.  An example of a symmetric configuration of thick wires is shown in \cref{fig_wires2d}.  As in 3D, one may truncate the distributions of infinite current slab or charged slab to a bounded region and apply suitable boundary charges or currents so as to maintain a linear interior field, such that the resulting distribution may be used to generate examples of symmetry promotion. 

Charge and current distributions in 2D have a duality which allows any example of symmetry promotion for the electric field to be reinterpreted as an example for the magnetic field and vice versa.  This duality follows directly from the Coulomb and Biot-Savart Laws in this case, 
\begin{align}
\E(\x) & = \frac{1}{4\pi\epsilon_0}\iiint_{\R^3} \rho(\x') \frac{\x-\x'}{\left|\x-\x'\right|^3} \intd^3\x' = \frac{1}{2\pi\epsilon_0}\iint_{\R^2} \rho(\x_\perp') \frac{\x_\perp-\x_\perp'}{\left|\x_\perp-\x_\perp'\right|^2} \intd^2\x_\perp' \\
\B(\x) & = \frac{\mu_0}{4\pi} \hat{\bm z}\times \iiint_{\R^3} J(\x') \frac{\x-\x'}{\left|\x-\x'\right|^3} \intd^3\x' = \frac{\mu_0}{2\pi} \hat{\bm z}\times \iint_{\R^2} J(\x_\perp') \frac{\x_\perp-\x_\perp'}{\left|\x_\perp-\x_\perp'\right|^2} \intd^2\x_\perp',
\end{align}
where we have used the $z$-translation invariance of $\rho$ and $J$ to evaluate the $z$ integral and $\x_\perp$ denotes the projection of $\x$ into the $x-y$ plane.  We see that the magnetic field $\B(\x)$ generated by current density $J(\x)$ is equal to $\hat{\bm{z}}\times \E(\x)/c$, where $\E(\x)$ is generated by charge density $\rho(\x) = J(\x)/c$.  Since the electric field is normal to equipotentials, we see that the dual magnetic field is everywhere tangent to the equipotential curves, which thus coincide with magnetic field lines.  We may therefore interpret examples like \cref{fig_slabs2d,fig_wires2d} as showing either charge distributions and equipotential curves or dual current distributions and magnetic field lines. 

\section{Generalizations and Discussion}
\label{sec_discussion}

We indicate briefly three ways in which the preceding results may be generalized.  Firstly, we discuss under what circumstances the role of the special orthogonal group $SO(n)$ can be replaced by the full orthogonal group $O(n)$, i.e. how do improper rotations figure into symmetry promotion?  Secondly, we mention some situations in which time dependent fields exhibit symmetry promotion.  Finally, we address the case of higher dimensional fields. 

\subsection{Improper Rotations}
In dimensions $n=2$ and $n=3$, it can be shown that any finite subgroup of $O(n)$ which acts irreducibly on $\R^n$ contains a subgroup consisting of proper rotations which also acts irreducibly on $\R^n$.  Thus if the symmetry of a charge or current distribution is invariant under some such group $G\subset O(n)$, the symmetry promotion results \cref{prop_3d} and \cref{prop_2d} still apply.  Since $O(n)$ is generated by $SO(n)$ and any element $g\in O(n) \setminus SO(n)$, if $G$ contains an improper rotation then the resulting field must be invariant under the full orthogonal group $O(n)$.  

One can show by direct calculation that a vector field of the form $\bm{v}(\x) = \lambda \x$ is always $O(n)$ invariant, while a pseudo-vector field of this form is only $O(n)$ invariant if it is trivial, i.e. $\lambda=0$.  To prove the latter statement, suppose that for some $R\in O(n)\setminus SO(n)$ and for all $\x \in \R^n$ we have 
\begin{equation}
\lambda \x = \bm{v}(\x) = \det(R) R \bm{v}(R^{-1} \x) = - R \lambda R^{-1} \x = -\lambda \x. 
\end{equation}  
Hence $\lambda=0$.  This provides an alternative proof that the only symmetry-promoted magnetic fields in 3D are uniformly zero.  Similarly, in 2D a vector field of the form $\bm{v}(\x) = \lambda \x + \mu Y\x$ is only $O(2)$ invariant for $\mu=0$.

\subsection{Time Dependent Fields}
A linearly varying electric field can be produced via Faraday's Law by ramping the current of a solenoid in time.  This electric field has the same functional form as the magnetic field from a thick wire.  Similarly, ramping the current of two parallel sheets can generate an electric field analogous to the magnetic field of a thick sheet.  In this way the 2D magnetic examples of \cref{fig_slabs2d,fig_wires2d} can be reinterpreted as electric examples in a way which is distinct from the duality mentioned in \cref{sec_symmetry_promotion_2d}.  Note in these cases the electric field so generated is not $O(2)$ invariant (see previous subsection), which reflects the fact that an inversion of space will reverse the flow of current, and with it the induced magnetic and electric fields. 

Linear magnetic fields can also in principle be constructed via Ampere's Law by ramping an electric field in time.  However, unlike the previous case, this does not lead to new geometries for symmetry promotion, since a time-varying electric field yields the same magnetic field as a current (in contrast to the previous case, where time-varying magnetic fields yield electric fields which are impossible to construct from static charge distributions). 

\subsection{Higher Dimensions}

We have shown in dimensions $n=2$ and $n=3$ that if an affine linear field is invariant under a group $G$ of proper rotations which acts irreducibly, then the field must necessarily be invariant under $SO(n)$.  This holds also in arbitrary odd dimensions $n$, since \cref{prop_3d} generalizes immediately to this case.  One might thus wonder if this statement of symmetry promotion holds in all dimensions.  Remarkably, this is not the case.  A counterexample is furnished by the quaternion group acting on $\R^4$.  The quaternion group generators $i, j, k$ can be represented by the matrices
\begin{equation}
\rep(i) = \begin{pmatrix} 0 & -1 & 0 & 0 \\ 1 & 0 & 0 & 0 \\ 0 & 0 & 0 & -1 \\ 0 & 0 & 1 & 0 \end{pmatrix} \qquad \rep(j) = \begin{pmatrix} 0 & 0 & -1 & 0 \\ 0 & 0 & 0 & 1 \\ 1 & 0 & 0 & 0 \\ 0 & -1 & 0 & 0 \end{pmatrix} \qquad \rep(k) = \begin{pmatrix} 0 & 0 & 0 & -1 \\ 0 & 0 & -1 & 0 \\ 0 & 1 & 0 & 0 \\ 1 & 0 & 0 & 0 \end{pmatrix}.
\end{equation}
We will verify below that the matrices $\pm \mathbbm{1}, \pm \rep(i), \pm \rep(j), \pm \rep(k)$ comprise an irreducible representation of the quaternion group.  Consider the $SO(4)$ matrices 
\begin{equation}
A = \begin{pmatrix} 0 & 1 & 0 & 0 \\ -1 & 0 & 0 & 0 \\ 0 & 0 & 0 & -1 \\ 0 & 0 & 1 & 0 \end{pmatrix} \qquad R = \begin{pmatrix} 0 & 0 & -1 & 0 \\ 0 & 1 & 0 & 0 \\ 0 & 0 & 0 & -1 \\ 1 & 0 & 0 & 0 \end{pmatrix} 
\end{equation}
and define a vector field $\bm{v}(\bm{x}) = A\bm{x}$. A straightforward calculation shows that $[A,g]=0$ for all $g\in \{\pm \mathbbm{1}, \pm \rep(i), \pm \rep(j), \pm \rep(k)\}$ and $[A,R] \neq 0$.  Thus $\bm{v}(\bm{x})$ provides a counterexample to symmetry promotion in 4D.

To show that $G = \{\pm \mathbbm{1}, \pm \rep(i), \pm \rep(j), \pm \rep(k)\}$ acts irreducibly on $\R^4$ it suffices to compute the quantities 
\begin{equation}
N = \frac{1}{|G|} \sum_{g\in G} \text{Tr}(g)^2 \qquad \qquad I = \frac{1}{|G|}\sum_{g\in G} \text{Tr}(g^2).
\end{equation}
The former quantity is the squared norm of the character of the representation, and the latter is known as the Frobenius-Schur indicator.  It can be shown using character theory \cite{fulton2013representation} that a real representation is irreducible if and only if $(N,I) = (1,1)$, $(N,I) = (2,0)$, or $(N,I) = (4,-2)$.  Direct computation shows that $(N,I) = (4,-2)$, hence the representation is irreducible. 

\section{Conclusion}
\label{sec_conclusion}
The objective of this work has been to see how far the initial example of \cref{fig_cylinders} can be generalized, and we have seen that it can be generalized quite far indeed.  The key ingredients for symmetry promotion are \textit{affine linearity} of the fields and \textit{irreducibility} of the group action.  Fields satisfying these hypotheses can look quite different in two vs. three dimensions.  The fact that symmetry promotion fails altogether in 4D in our opinion makes the phenomenon all the more interesting in those dimensions for which it does occur. Another common feature to all our examples is the presence of sharp boundaries in the source distribution, which divide space into an interior region where symmetry promotion applies and an exterior region where the fields only exhibit the smaller symmetry of the source distribution. 

We view much of the value of these results as pedagogical in nature.  Symmetry promotion in electro- and magnetostatics provides a wealth of examples in an area known for having a limited sample of solvable problems.  It also is an accessible case study in group representation theory, which stands out all the more for involving non-trivial subtleties of representations on real vector spaces. 

\begin{acknowledgments}
We thank Tadashi Tokieda for helpful discussions, and for pointing out an error in an earlier draft of the manuscript.  The authors acknowledge use of Gemini Pro 3.1 for construction the symmetry promotion counterexample in 4D, as well as OpenAI Codex and ChatGPT for error checking of the manuscript. 
\end{acknowledgments}

\bibliography{bibs}

\end{document}